\documentclass[conference]{IEEEtran}
\usepackage{cite} 
\usepackage{units}

\usepackage{amsmath} 
\usepackage{amssymb,amsthm} 
\usepackage[subnum]{cases}
\usepackage{ctable}
\usepackage{dsfont}
\usepackage{graphicx,graphics,color}
\graphicspath{{figures/}} 

\usepackage[caption=false,font=footnotesize]{subfig}

\newcommand{\Pul}{P_{\text{UL}}}
\newcommand{\Pc}{P_{\text{C}}}
\newcommand{\hdl}{h_{\text{DL}}}
\newcommand{\hul}{h_{\text{UL}}}
\newcommand{\Omegaul}{\Omega_{\text{UL}}}
\newcommand{\Omegadl}{\Omega_{\text{DL}}}
\newcommand{\Pdl}{P_{\text{DL}}}

\newcommand{\dd}{\text{d}}
\newcommand{\E}{\mathbb{E}}
\newcommand{\e}{{\rm e}}
\newcommand{\req}{\overset{\underset{!}{}}{=}}

\theoremstyle{plain}
\newtheorem{corollary}{Corollary}
\newtheorem{theorem}{Theorem}

\newtheorem{proposition}{Proposition}

\theoremstyle{definition}

\theoremstyle{remark}

\begin{document}
\title{On-Off Transmission Policy for Wireless Powered Communication with Energy Storage}

\author{\IEEEauthorblockN{Rania Morsi, Diomidis S. Michalopoulos, and Robert Schober}
\IEEEauthorblockA{Institute of Digital Communications, Friedrich-Alexander-University Erlangen-N\"urnberg (FAU), Germany}
}


\maketitle

\begin{abstract}
In this paper, we consider an energy harvesting (EH) node which harvests energy from a radio frequency (RF) signal broadcasted by an access point (AP) in the downlink (DL). The node stores the harvested energy in an energy buffer and uses the stored energy to transmit data to the AP in the uplink (UL). We consider a simple transmission policy, which accounts for the fact that, in practice, the EH node may not have knowledge of the EH profile nor of the UL channel state information. In particular, in each time slot, the EH node transmits with either a constant desired power or remains silent if not enough energy is available in its energy buffer. For this simple policy, we use the theory of discrete-time continuous-state Markov chains to analyze the limiting distribution of the stored energy for finite- and infinite-size energy buffers. Moreover, we take into account imperfections of the energy buffer and the circuit power consumption of the EH node. For a Rayleigh fading DL channel, we provide the limiting distribution of the energy buffer content in closed form. In addition, we analyze the average error rate and the outage probability of a Rayleigh faded UL channel and show that the diversity order is not affected by the finite capacity of the energy buffer.  Our results reveal that, for medium to high signal-to-noise ratio (SNRs), the optimal target transmit power of the EH node is less than the average harvested power and increases with the capacity of the energy buffer.
\end{abstract}

\section{Introduction}
The performance of battery-powered wireless communication networks, such as sensor networks, is limited by the lifetime of the network nodes. Periodic replacement of the nodes' batteries is costly, inconvenient, and sometimes impossible. The lifetime bottleneck problem of energy-constrained wireless networks thus demands harvesting energy from renewable energy sources (e.g., solar, wind, thermal, vibration) to ensure a sustainable network operation. The harvested energy can then be used by the energy harvesting (EH) node to transmit data to its designated receiver. However, the aforementioned energy sources are intermittent and uncontrollable. For example, solar and wind energy are weather dependent and are not available indoors. In contrast, radio frequency (RF) energy is a viable energy source which is partially controllable and can be provided on demand to charge low-power devices \cite{Kansal_2007}. 

A common feature of EH communication networks is the randomness of the amount of harvested energy and the randomness of the fading affecting the information link. 
Therefore, one main objective of energy management polices for EH networks is to match the energy consumption profile of the EH node to the random energy generation profile of the EH source and to the random information channel \cite{Kansal_2007,WPC_TDMA,WPC_SDMA,Sharma2010,Ulukus_Yener_2011,RuiZhang2012,Outage_RuiZhang_2012}. 
In \cite{WPC_TDMA,WPC_SDMA}, a harvest-then-transmit protocol is considered for a multiuser system with RF wireless power transfer (WPT) in the downlink (DL) and wireless information transfer (WIT) in the uplink (UL), where the users' sum rate or equal throughput is maximized on a per-slot basis. In \cite{Sharma2010}, throughput and mean delay optimal transmission policies, which stabilize the data queue of an EH sensor node over an infinite horizon, are proposed  in a time-slotted setting. In \cite{Ulukus_Yener_2011}, optimal transmission policies that maximize the throughput by a deadline or minimize the transmission completion time are proposed for an EH node with finite energy storage in a continuous time setting. Optimal power allocation solutions that maximize the throughput and minimize the outage probability over a finite horizon for an EH node with infinite energy storage in a time-slotted setting are reported in \cite{RuiZhang2012} and \cite{Outage_RuiZhang_2012}, respectively. 

Optimal offline transmission policies typically require non-causal knowledge of energy and channel state information (CSI) at the EH node, whereas optimal online solutions are typically based on dynamic programming which is computationally intensive even for a small number of transmitted symbols, see  \cite{Ulukus_Yener_2011,RuiZhang2012} and references therein. Therefore, these optimal policies may not be feasible in practice. For example, typical EH wireless sensor networks are expected to comprise many small, inexpensive sensors with limited computational power and energy storage. In such networks, even causal CSI may not be available at the EH nodes nor at the energy source.

Motivated by these practical considerations, in this paper, we consider a simple online transmission policy, where the CSI and the EH profile are not available at the EH node nor at the energy source. In particular, an access point (AP) transmits an RF signal with a constant power in the DL and the EH node harvests the received RF energy and uses the stored energy to transmit data to the AP in the UL. We consider an on-off transmission policy, where in each time slot, the EH node transmits with either a constant desired power or remains silent if the stored energy can not support the desired UL transmit power. This is unlike the transmission policy considered in \cite{Morsi_ICC2014}, where the output power is either a constant desired power or a lower power if not enough energy is available in the energy buffer. The on-off transmission policy considered in this paper is motivated by the fact that a constant transmit power allows the use of efficient power amplifiers at the EH node. We model the stored energy by a discrete-time continuous-state Markov chain and provide its limiting distribution for both infinite and finite energy storage, when the DL channel is Rayleigh fading. Under this framework, we analyze the average error rate (AER) and the total outage probability (which includes both outages due to missed transmission opportunities as well as channel outages) of a Rayleigh fading information link. We show that, surprisingly, the diversity order of the total outage probability  is independent of the storage capacity. Furthermore, we show that, except for high outage probabilities $(>0.5)$, the optimal desired UL power which minimizes the total outage probability is always less than the average harvested power and increases with the capacity of the energy buffer. The proposed framework also takes into account system non-idealities such as non-zero circuit power consumption and imperfections of the energy buffer. 

\section{System Model}
\label{s:System_model}
We consider a time-slotted point-to-point single-antenna EH system with DL WPT and UL WIT. In particular, the system consists of a node with an EH module which captures the RF energy transferred by an AP in the DL and uses the harvested energy to transmit its backlogged data in the UL. The considered system employs frequency-division-duplex, where WPT and WIT take place concurrently on two different frequency bands. The AP and the EH node are assumed to not have knowledge of the instantaneous DL and UL CSI, nor of the amount of harvested energy. Next, we describe the communication, EH, and storage models as well as the considered system imperfections.
\subsection{Communication Model}
In time slot $i$ (defined as the time interval $[i,i+1)$\footnote{The time slot is assumed to be of unit length. Hence, we use the terms energy and power interchangeably.}), the EH node transmits data to the AP with an UL power given by
\begin{equation}
\Pul(i)=M\mathds{1}_{B(i)>M}=\begin{cases} 0 & B(i)\leq M \\ M & B(i) > M\end{cases},
\label{eq:Pul_policy}
\end{equation}
where $\mathds{1}_A=1$ if event $A$ is true, and $\mathds{1}_A=0$ otherwise. $B(i)$ is the residual stored energy at the beginning of time slot $i$ and $M$ is the desired constant UL transmit power. The UL channel is assumed to be flat block fading, i.e., the channel remains constant over one time slot, and changes independently from one slot to the next. The channel power gain sequence $\{\hul(i)\}$ is a stationary and ergodic process with mean $\Omegaul=\E[\hul(i)]$, where $\E[\cdot]$ denotes expectation. Additive white Gaussian noise (AWGN) with variance $\sigma_n^2$ impairs the received signal at the AP.
\subsection{EH Model}
In time slot $i$, the EH node also collects $X(i)$ units of RF energy broadcasted by the AP and stores it in its energy buffer.  We assume that the energy replenished in a time slot may only be used in future time slots. The DL channel is also assumed to be flat block fading with  a stationary and ergodic channel power gain sequence $\{\hdl(i)\}$,  assumed to be unknown at the AP, and $\Omegadl=\E[\hdl(i)]$. We adopt the EH receiver model in \cite{WIPT_Architecture_Rui_Zhang_2012}, where the harvested energy in time slot $i$ is given by $X(i)=\eta\Pdl\hdl(i)$, where $0<\eta< 1$ is the RF-to-DC conversion efficiency of the EH module and $\Pdl$ is the constant DL transmit power of the AP. The energy replenishment sequence $\{X(i)\}$ is an independent and identically distributed (i.i.d.) stationary and ergodic process with mean $\bar{X}=\eta\Pdl\Omegadl$, probability density function (pdf) $f(x)$, and complementary cumulative distribution function (ccdf) $\bar{F}(x)=\mathbb{P}(X(i)>x)$, where $\mathbb{P}(\cdot)$ denotes the probability of an event.
\subsection{Storage Model}
\label{ss:storage_model}
The harvested energy $X(i)$ is stored in an energy buffer, such as a rechargeable battery or a supercapacitor, with storage capacity $K$. The dynamics of the storage process $\{B(i)\}$ are given by the storage equation
\begin{equation}
\begin{aligned}
B(i+1)&=\min\left(B(i)-\Pul(i)+X(i),K\right),\\
&=\min\left(B(i)-M\mathds{1}_{B(i)> M}+X(i),K\right).\\
\end{aligned}
\label{eq:general_storage_equation}
\end{equation} 
The storage process $\{B(i)\}$  in (\ref{eq:general_storage_equation}) is a discrete-time Markov chain on a continuous state space $S$, where  $S=[0,K]$ and  $S=[0,\infty)$ for finite- and infinite-size energy buffers, respectively. 
\subsection{Consideration of Imperfections}
\label{s:imperfections}
We consider imperfections due to the power consumed by the EH node circuitry during transmission and the non-idealities of the energy buffer. In particular, we consider the following imperfections: (a) To produce an RF power of $\Pul$, the power amplifier of the EH node consumes a total power of $\alpha\Pul$, where $\alpha> 1$ is the power amplifier inefficiency. (b) We assume that the circuitry of the EH node consumes a constant power of $\Pc$ during transmission. (c) The energy buffer is characterized by a storage efficiency $0<\beta<1$, where if $X$ amount of energy is applied at the input of the buffer, only an amount of $\beta X$ is stored. 

Consequently, the energy buffer dynamics are described by 
\begin{equation}
B(i+1)\!=\!\min\left(B(i)\!-\!(\Pc+\alpha M)\mathds{1}_{B(i)>\Pc+\alpha M}+\beta X(i),K\right)
\label{eq:storage_equation_imperfection}
\end{equation}
Observe that (\ref{eq:storage_equation_imperfection}) is identical to (\ref{eq:general_storage_equation}) after replacing $M$ by $\tilde{M}\!=\!\Pc+\alpha M$ and $f(x)$ by $\tilde{f}(x)\!=\!\frac{1}{\beta}f\left(\frac{x}{\beta}\right)$. Thus, in the following, we perform the analysis for an ideal system (i.e., $\alpha\!=\!1$, $\beta\!=\!1$, and $\Pc\!=\!0$). For a non-ideal system, all the results in Sections \ref{s:Infinite_buffer}-\ref{s:BER_outage_analysis} hold with the aforementioned substitutions.
\section{Infinite-Capacity Energy Buffer}
\label{s:Infinite_buffer}
In this section, we study the energy storage process in (\ref{eq:general_storage_equation}) for an infinite-capacity energy buffer. We provide conditions for which the convergence to a limiting distribution of the buffer content is either guaranteed or violated. Furthermore, we provide the limiting distribution of the buffer content in closed form when the EH process $\{X(i)\}$ is i.i.d. exponentially distributed,  i.e., for a Rayleigh block fading DL channel.
\begin{theorem}\normalfont
For the storage process $\{B(i)\}$ in (\ref{eq:general_storage_equation}) with infinite buffer size, if $M<\bar{X}$, then $\{B(i)\}$ does not possess a stationary distribution\footnote{A stationary distribution of a Markov chain is a distribution such that if the chain starts with this distribution, it remains in this distribution.}. Furthermore, after a finite number of time slots, $\Pul(i)=M$ holds almost surely (a.s).
\label{theo:no_stationary_dist}
\end{theorem}
\begin{proof} The proof is identical to that of \cite[Theorem 1]{Morsi_ICC2014}. \end{proof}
\begin{theorem}\normalfont
For the storage process $\{B(i)\}$ in (\ref{eq:general_storage_equation}) with infinite buffer size, if $M>\bar{X}$, then $\{B(i)\}$ is a stationary and ergodic process which possesses a unique stationary distribution $\pi$ that is absolutely continuous on $(0,\infty)$. Furthermore, the process converges in total variation to the limiting distribution\footnote{A limiting distribution of a Markov chain is a stationary distribution that the chain converges to asymptotically from some initial distribution.} $\pi$ from any initial distribution. 
\label{theo:stationary_dist_infinite}
\end{theorem}
\begin{proof} The proof will be provided in the journal version of this paper. \end{proof}\vspace{-0.1cm}
\begin{theorem}\normalfont
Consider the storage process $\{B(i)\}$ in (\ref{eq:general_storage_equation}) with infinite buffer size and $M>\bar{X}$. Let $g(x)$ on $(0,\infty)$ be the limiting pdf of the energy buffer content, then $g(x)$ must satisfy the following integral equations\vspace{-0.1cm}
\begin{numcases}{g(x)\!=\!\! \label{eq:g_integral_eqn_infinite}}
\hspace{-0.1cm}\int\limits_{u=0}^x f(x-u) g(u) \dd u + \int\limits_{u=M}^{M+x} f(x-u+M) g(u) \dd u, & \nonumber \\ \vspace{-0.2cm}
\hspace{4.4cm}0\leq x< M & \hspace{-1cm}\label{eq:Integral_eqn_infinite_parta}\\\vspace{0.1cm}
\hspace{-0.1cm}\int\limits_{u=0}^M f(x-u)g(u) \dd u + \int\limits_{u=M}^{M+x} f(x-u+M) g(u) \dd u,  & \nonumber \\
\hspace{4.4cm} x\geq M  & \hspace{-1cm}\label{eq:Integral_eqn_infinite_partb}\end{numcases}
\end{theorem}
\begin{proof}
To understand the integral equation in (\ref{eq:g_integral_eqn_infinite}), one may set $B(i)=u$ and $B(i+1)=x$, then (\ref{eq:general_storage_equation}) reads\vspace{-0.1cm}
\begin{equation}
x=\begin{cases} u+X(i), & u\leq M \\ u-M+X(i), & u>M \end{cases}.\vspace{-0.1cm}
\label{eq:Storage_eq_Infinite_u_x}
\end{equation}
Thus, for $u\leq M$, $g(x|u\leq M)=f(x-u)$ which is non-zero only for a non-negative amount of harvested energy, i.e., for $u\leq x$. Hence, in the range $u\leq M$, the upper limit on $u$ in the first integral of (\ref{eq:Integral_eqn_infinite_parta}) and (\ref{eq:Integral_eqn_infinite_partb}) is given by $u=\min(x,M)$. The second integral in (\ref{eq:Integral_eqn_infinite_parta}) and (\ref{eq:Integral_eqn_infinite_partb}) corresponds to the range $u>M$, where from (\ref{eq:Storage_eq_Infinite_u_x}), $g(x|u> M)=f(x-u+M)$ which is non-zero only for a non-negative amount of harvested energy, i.e., for $u\leq M+x$. This completes the proof. 
\end{proof}\vspace{-0.1cm}
Next, we consider the case when the DL channel is Rayleigh block fading and provide the limiting distribution of the energy buffer content in the following corollary.
\begin{corollary} \normalfont
Consider the storage process in (\ref{eq:general_storage_equation}) with infinite buffer size and $M\!>\!\bar{X}$. If the EH process is exponentially distributed with pdf $f(x)\!=\!\lambda \e^{-\lambda x}$, where $\lambda\!\!=\!\!\frac{1}{\bar{X}}$  and $\delta\!=\!\lambda M\!=\!\frac{M}{\bar{X}}$, then the limiting pdf of the energy buffer content is given by
\begin{equation} g(x)=
\begin{cases}
\frac{1}{M}\left(1-\e^{px}\right), & 0<x\leq M\\
\frac{-p}{M(\lambda+p)}\e^{px}=\frac{-p}{\delta\e^{pM}}\e^{px}, & x>M
\end{cases}
\label{eq:g_x_infinite_buffer}
\end{equation}
where $p\!<\!0$ satisfies $\lambda\e^{pM}\!=\!\lambda+p$ and is given by $p\!=\frac{-\delta-W_{0}(-\delta\e^{-\delta})}{M}$. Here, $W_0(\cdot)$ is the Lambert W function of order zero. Furthermore, the transmission probability is given by $\mathbb{P}(\Pul(i)=M)=\frac{1}{\delta}$.
\label{theo:stationary_dist_infinite_exp}
\end{corollary} 
\begin{proof} The proof is provided in Appendix \ref{app:stationary_dist_infinite_exp}. \end{proof}\vspace{-0.2cm}
\section{Finite-Capacity Energy Buffer}
\label{s:Finite_buffer}
In this section, we first provide the integral equation of the stationary distribution of the storage process $\{B(i)\}$ for a finite-size energy buffer and a general i.i.d. EH process. Then, the distribution is provided in closed form for a Rayleigh fading DL channel.
\begin{theorem}  \normalfont
For a finite buffer size $K$ and an EH process $\{X(i)\}$ with a distribution that has an infinite positive tail, the storage process in (\ref{eq:general_storage_equation}) is a stationary and ergodic process which possesses a unique stationary distribution $\pi$ that has a density on $(0,K)$ and an atom at $K$. Furthermore, the process converges in total variation to the limiting distribution $\pi$ from any initial distribution. 
\label{theo:limiting_dist_Finite}
\end{theorem}
\begin{proof} The proof is identical to that of \cite[Theorem 4]{Morsi_ICC2014}. \end{proof}
\begin{theorem}  \normalfont
Consider the storage process $\{B(i)\}$ in (\ref{eq:general_storage_equation}), with a finite buffer size $K>2M$. Let $g(x)$ be the limiting pdf of the energy buffer content on $(0,K)$ and $\pi(K)$ be the limiting probability of a full buffer (i.e., the atom at $K$). If $f(x)$ and $\bar{F}(x)$ are respectively the pdf and the ccdf of $\{X(i)\}$, then, $g(x)$ and $\pi(K)$ must jointly satisfy\vspace{-0.1cm}
\begin{numcases}{g(x)\!=\!\! \label{eq:g_integral_eqn_finite}}
\hspace{-0.1cm}\int\limits_{u=0}^{x}f(x-u) g(u) \dd u + \int\limits_{u=M}^{M+x} f(x-u+M) g(u)\dd u, &  \nonumber \\ \vspace{-0.2cm}
\hspace{5cm}0\leq x<M & \label{eq:parta}\\
\hspace{-0.1cm}\int\limits_{u=0}^{M}f(x-u) g(u) \dd u + \int\limits_{u=M}^{M+x} f(x-u+M) g(u)\dd u, &  \nonumber \\ \vspace{-0.2cm}
\hspace{4.4cm}M\leq x<K-M & \label{eq:partb}\\
\hspace{-0.1cm}\int\limits_{u=0}^{M}f(x-u) g(u) \dd u + \int\limits_{u=M}^{K} f(x-u+M) g(u)\dd u & \nonumber \\
+\pi(K) f(x-K+M),  \hspace{1.1cm} K-M \leq x<K & \label{eq:partc}
\end{numcases}
\begin{equation}\vspace{-0.2cm}
\pi(K)\!=\!\frac{\left[\int\limits_{u=0}^{M}\bar{F}(K-u) g(u) \dd u + \!\!\int\limits_{u=M}^{K} \bar{F}(K-u+M) g(u)\dd u \right]}{1-\bar{F}(M)},\vspace{-0.1cm}
\label{eq:partd}
\end{equation}
and the unit area condition\vspace{-0.1cm}
\begin{equation}
\int\limits_{0}^{K} g(u) \dd u + \pi(K)=1.\vspace{-0.1cm}
\label{eq:unit_area_eq}
\end{equation}
\end{theorem}
\begin{proof}
The integral equations in (\ref{eq:g_integral_eqn_finite}) and (\ref{eq:partd}) can be derived by adopting the same approach used for the proof of (\ref{eq:g_integral_eqn_infinite}). In particular, if we set $B(i)=u$ and $B(i+1)=x$, then (\ref{eq:general_storage_equation}) reads
\begin{equation}
x=
\begin{cases}
u+X(i) & u\leq M \quad\&\quad u+X(i)<K\\
u-M+X(i) & u>M \quad\&\quad u-M+X(i)<K\\
K & {\rm otherwise.}
\end{cases}
\label{eq:Our_storage_equation_cases}
\end{equation}
Consider first the continuous part of the distribution, i.e., $g(x)$ defined on $0\leq x <K$ given in (\ref{eq:g_integral_eqn_finite}). Eqs. (\ref{eq:parta}) and (\ref{eq:partb}) are identical to (\ref{eq:Integral_eqn_infinite_parta}) and (\ref{eq:Integral_eqn_infinite_partb}), respectively. However, we need to further ensure that the upper limit on $u$ given by $M+x$ (for a non-negative harvested energy) is in the domain of $g(u)$. That is, in (\ref{eq:parta}), $\max_x(M+x)<K$ must hold, i.e., $K>2M$ and in (\ref{eq:partb}), $M+x<K$ must hold. Hence, (\ref{eq:partb}) is valid only for $x<K-M$ (with strict inequality). For the rest of the range of $x$ in (\ref{eq:partc}), i.e., $K-M\leq x<K$, the upper limit  $M+x$ on $u$ is larger than or equal to $K$. Thus, the whole range of $0< u\leq K$ contributes to $g(x)$. The range $0< u<K$ is covered by the first two integrals in (\ref{eq:partc}), and $u=K$ is considered in the last term. Finally, at $x=K$, the probability that the buffer is full, $\pi(K)$, in (\ref{eq:partd}) is obtained similar to (\ref{eq:partc}). However, rather than considering the pdf at the amount of harvested energy $x-u+M\mathds{1}_{u>M}$ as in (\ref{eq:partc}), we consider the ccdf $\bar{F}(x-u+M\mathds{1}_{u>M})$ instead (at $x\!=\!K$). This is because the full buffer level $K$ is attained when the amount of harvested energy is larger than or equal to $K-u+M\mathds{1}_{u>M}$, where we sweep over $0\!<\!u\!\leq\!K$ to obtain (\ref{eq:partd}).
\end{proof}
Next, we consider the case when the DL channel is Rayleigh block fading. We provide the exact limiting distribution of the energy buffer content in Corollary \ref{theo:limiting_dist_Finite_exponential_exact} and an approximate distribution in Proposition \ref{prop:limiting_dist_Finite_exponential_approx}.
\begin{corollary}  \normalfont
Consider the storage process $\{B(i)\}$ in (\ref{eq:general_storage_equation}) with a finite buffer size $K=lM$, with $l\in\mathbb{Z}$ and $l\geq 2$, and an i.i.d. exponentially distributed EH process $\{X(i)\}$ with pdf $f(x)=\lambda\e^{-\lambda x}$, where $\lambda=\frac{1}{\bar{X}}$ and $\delta=\lambda M$, then the limiting pdf $g(x)$ of the energy buffer content and the full buffer probability $\pi(K)$ satisfy
 \small
\begin{numcases}{g(x)\!=\!\!\label{eq:g_x_finite_exact}}
\begin{aligned}
&\pi(K)\lambda\e^{-\delta(1-l)} \sum\limits_{q=0}^{l-2}\frac{\e^{-\delta q}}{q!} \Big[\left(\delta (q+1)-\lambda K\right)^{q}\\
&-\e^{-\lambda x}\left(\delta (q+1)+\lambda(x-K)\right)^{q}\Big],\end{aligned} & \nonumber\\
\hspace{4.9cm} 0\leq x<M,\label{eq:g_x_finite_exact_parta} \hspace{-1cm}& \\
\begin{aligned}
&\pi(K)\lambda\e^{-\lambda(x-K)} \Bigg[1+\sum\limits_{q=1}^{n}\frac{\e^{-\delta q}}{(q-1)!} \left(\delta q+\lambda(x-K)\right)^{q-1}\\
&+\left(\frac{\lambda(x-K)}{q}+\delta-1\right)\Bigg],\end{aligned} &\nonumber\\%
\hspace{2cm} \begin{aligned}&K-(n+1)M\leq x<K-nM,\\ &\hspace{2.1cm} n=0,\ldots,l-2, \end{aligned}\label{eq:g_x_finite_exact_partb} \hspace{-0.2cm}& 
\end{numcases}
\normalsize \vspace{-0.1cm}
and \vspace{-0.1cm}
 \small
\begin{equation}
\begin{aligned}
\pi(K)\!=\!\!\Bigg\{&1+\!\sum\limits_{n=0}^{l-2}\e^{n\delta}\sum\limits_{q=0}^{n}\!\frac{\left(\delta\e^{-\delta}\right)^q}{q!}\!\left(\!\e^{\delta}\left(q\!-\!(n\!+\!1)\right)^q\!-\!(q\!-\!n)^q\right)\! \\
&+\e^{-\delta(1-l)}\sum\limits_{q=0}^{l-2}\frac{\e^{-\delta q}}{q!}\Big[\e^{\delta(q+1-l)}\big(\Gamma\left(q+1,\delta(q+2-l)\right)\\
&-\Gamma\left(q+1,\delta(q+1-l)\right)\big)+\delta\left(\delta(q+1-l)\right)^q\Big]\Bigg\}^{-1}.
\end{aligned}
\label{eq:Pi_K_delta}
\end{equation}
\normalsize
\label{theo:limiting_dist_Finite_exponential_exact}  
\end{corollary}\vspace{-0.4cm}
\begin{proof} 
The exact derivation of the limiting distribution is lengthy so we omit it and provide it in the journal version of this paper. Here, we describe the approach used to derive (\ref{eq:g_x_finite_exact}) and (\ref{eq:Pi_K_delta}). First, observe that (\ref{eq:partb}), (\ref{eq:partc}), and (\ref{eq:partd}) are identical to \cite[(8a), (8b), and (9)]{Morsi_ICC2014}, respectively, except for the first term on the right hand side of the three equations. The solution for $g(x)$ in \cite{Morsi_ICC2014} was obtained recursively backwards in stripes of width $M$ \emph{independent} of the value of the first term in \cite[(8a), (8b), and (9)]{Morsi_ICC2014}. As a result,  in the range $M\leq x <K$, the limiting pdf $g(x)$ for the transmission policy considered in this paper in (\ref{eq:Pul_policy}) is identical to that for the transmission policy $\Pul(i)=\min(B(i),M)$ considered in \cite{Morsi_ICC2014}. That is, (\ref{eq:g_x_finite_exact_partb}) is identical to \cite[eq. (12)]{Morsi_ICC2014}. Compared to \cite{Morsi_ICC2014}, the limiting distribution of the considered policy only differs in $g(x)$, $0 \leq x <M$, and consequently in the value of the atom $\pi(K)$. Now, define the $n^{{\rm th}}$ section of $g(x)$ as \vspace{0.1cm}  
\begin{equation} 
g_n(x)=g(x),\quad K-(n+1)M\leq x<K-nM
\label{eq:g_n_x}
\end{equation}\vspace{0.1cm}
and use  $f(x)=\lambda\e^{-\lambda x}$ and $\bar{F}(x)=\e^{-\lambda x}$, then $g(x)$, $0\leq x<M$, is obtained from (\ref{eq:parta}) as
\small
\begin{equation}
g_{l-1}(x)\!=\!\lambda\e^{-\lambda x}\!\!\int\limits_{u=0}^{x}\!\e^{\lambda u}g_{l-1}(u) \dd u + \underbrace{\lambda\e^{-\lambda x}\!\!\int\limits_{u=M}^{M+x}\!\e^{-\lambda(M-u)} g_{l-2}(u)\dd u}_{r(x)},
\label{eq:g_l_1_integral_eq}
\end{equation}
\normalsize
where $g_{l-2}(x)$ is given by (\ref{eq:g_x_finite_exact_partb}) at $n=l-2$. The integral equation in (\ref{eq:g_l_1_integral_eq}) is a Volterra integral equation of the second kind whose solution is given in \cite[eq. 2.2.1]{polyanin2008handbook} as $g_{l-1}(x)=r(x)+\lambda\int\limits_0^x r(t) \dd t$, which reduces to (\ref{eq:g_x_finite_exact_parta}). Finally, $\pi(K)$ in (\ref{eq:Pi_K_delta}) guarantees a unit area distribution, i.e.,\vspace{0.1cm}
\begin{equation}
\int\limits_0^M g_{l-1}(x) \dd x + \sum\limits_{n=0}^{l-2}\,\int\limits_{K-(n+1)M}^{K-nM}g_n(x) \dd x + \pi(K)=1.
\label{eq:PiK_exact_Transmit_M_only_step}
\end{equation}\vspace{0.1cm}
We note that the limiting distribution provided in Corollary \ref{theo:limiting_dist_Finite_exponential_exact} is valid only for $K=lM$, where $l\in\mathbb{Z}$ and $l\geq 2$. The solution for a general buffer size is more involved and will be provided in the journal version of this paper.
\end{proof}\vspace{-0.2cm}
Since the exact limiting distribution of the buffer content provided in Corollary \ref{theo:limiting_dist_Finite_exponential_exact} is quite complicated, we propose an approximate distribution which will be used in Section \ref{s:BER_outage_analysis}.
\begin{proposition}  \normalfont
The limiting distribution of the storage process described in Corollary \ref{theo:limiting_dist_Finite_exponential_exact} can be approximated in the range $0 \leq x <M$ by $\tilde{g}_{l-1}(x)$, and in the range $M\leq x < K-n_c M$ by $\tilde{g}(x)$, where $n_c$ is some chosen section number such that $g_n(x)\approx \tilde{g}(x),\,\forall\, n_c\leq n\leq l-2$ as shown in Fig. \ref{fig:pdf_approximation}. For $k-n_cM \leq x < K$, $\tilde{g}_n(x)$ is given by the exact $g_n(x)$ in (\ref{eq:g_x_finite_exact_partb}) after replacing $\pi(K)$ with $\tilde{\pi}(K)$, where $\tilde{\pi}(K)$ is the approximate probability that the buffer is full which ensures a unit area of the approximate distribution. The proposed approximation is tight for $K\geq 3M$ and $n_c\geq 2$. Furthermore, the approximation error tends to zero as the buffer size tends to infinity (for $\delta=M/\bar{X}>1$).
\vspace{-0.3cm}
\begin{figure}[!h]
\centering
\scalebox{0.57}{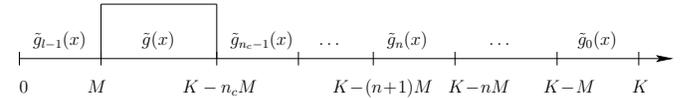}\vspace{-0.3cm}
\caption{Pdf approximation.}
\label{fig:pdf_approximation}
\end{figure}\newline
The approximate pdf is given by
\small
\begin{align}
&\tilde{g}_{l-1}(x)=\begin{cases} c\frac{\lambda\e^{dM}}{d}\left[\e^{dx} -1 \right] & \delta \neq 1\\
 c\lambda x & \delta=1 \end{cases} &\hspace{0.2cm} & 0\leq x<M\label{eq:gtilde_l_1}\\
&\tilde{g}(x)=c\e^{dx}  & & M\leq x<K-n_cM\\
&\tilde{g}_n(x)=\frac{\tilde{\pi}(K)}{\pi(K)}g_n(x) &&\hspace{-1cm} \begin{aligned} K\!-\!(n\!+\!1)&M\leq x<K\!-\!nM\\ &n=0,\cdots n_c-1\end{aligned}
\end{align}
\normalsize
where $d$ and $c$ are given by \vspace{-0.2cm}
\begin{equation}
d=\frac{-\delta-W_j(-\delta\e^{-\delta})}{M},\quad j=\begin{cases} -1 & 0<\delta\leq1\\
																																		0 & \delta>1
																																\end{cases},
\label{eq:d_approx}
\end{equation} \vspace{-0.2cm}
\small
\begin{equation*}
c=\tilde{\pi}(K)\,\lambda\underbrace{\e^{\lambda K} \left[1\!+\!\sum\limits_{q=1}^{l-1}\frac{\e^{-\delta q}}{(q-1)!} \left(\delta q\!-\!\lambda K\right)^{q-1}\left(\frac{-\lambda K}{q}\!+\!\delta\!-\!1\right)\right]}_{\Sigma_1}.
\label{eq:c_approx}
\end{equation*}
\normalsize
Here, $W_j(\cdot)$ is the $j^{\text{th}}$ order Lambert W function and the approximate atom at $K$ is given by
\small
\begin{equation}
\tilde{\pi}(K)\!=\!\!\Bigg[\!1\!+\!\Sigma_2\!+\!\!\sum\limits_{n=0}^{n_c-1}\!\!\e^{n\delta}\!\sum\limits_{q=0}^{n}\!\frac{\left(\delta\e^{-\delta}\right)^q}{q!}\!\left(\e^{\delta}\left(q\!-\!(n\!+\!1)\right)^q\!-\!(q\!-\!n)^q\!\right)\!\Bigg]^{-1},
\label{eq:Pi_K_approx}
\end{equation}
\normalsize
where $\Sigma_2=\begin{cases} \frac{\lambda\Sigma_1}{d}\left(\e^{d(K-n_c M)}-\delta \e^{dM}\right) & \delta\neq 1\\
\lambda\Sigma_1\left(\frac{M}{2}+K-(n_c+1) M\right) & \delta=1 \end{cases}$.
\label{prop:limiting_dist_Finite_exponential_approx}
\end{proposition}
\begin{proof} The proof is provided in Appendix \ref{app:limiting_dist_Finite_exponential_approx}.\end{proof}
\section{AER and Outage Probability Analysis}
\label{s:BER_outage_analysis}
In this section, we analyze the AER and the outage probability of the communication over the UL channel, when both the UL and the DL channels are Rayleigh faded.
With the considered on-off transmission policy in (\ref{eq:Pul_policy}), whenever the EH node transmits, its UL transmit power is constant. Hence, the AER and the channel outage probability are 
independent of the energy buffer size. On the other hand, the transmission probability depends on the buffer size. For an infinite-size buffer, the transmission probability is given by $\mathbb{P}_{\rm trans}=\mathbb{P}(\Pul(i)=M)=1$ if $\delta\leq 1$, c.f. Theorem \ref{theo:no_stationary_dist} and by $\mathbb{P}_{\rm trans}=1/\delta$ if $\delta>1$, c.f. Corollary \ref{theo:stationary_dist_infinite_exp}. For a finite-size buffer,  the transmission probability is given by
\begin{equation}
\begin{aligned}
&\mathbb{P}_{\rm trans}=\mathbb{P}(B(i)>M)\\
&=1-\int\limits_{0}^M \tilde{g}_{l-1}(x)=\begin{cases}1-\frac{c\e^{dM}}{d}(1-\delta) & \delta \neq 1\\ 1-c\frac{M}{2} & \delta=1 \end{cases},
\end{aligned}
\label{eq:Pr_transmission_Finite}
\end{equation}
where we used the approximate pdf in (\ref{eq:gtilde_l_1}) and $\lambda+d=\lambda\e^{dM}$. 
\subsection{AER Analysis}
\label{ss:BER_Analysis}
For uncoded transmission, the bit or symbol error rate of many coherent modulation schemes can be expressed as $P_e(\gamma)=aQ(\sqrt{b\gamma})$ \cite{Wang_Giannakis_2003}, where $\gamma$ is the instantaneous SNR, $Q(\cdot)$ is the Gaussian Q-function, and $a$ and $b$ depend on the modulation scheme used, e.g., for binary phase shift keying (BPSK), $a\!=\!1$ and $b\!=\!2$. 
Hence, the AER is given by 
\begin{equation}
P_e=\int\limits_0^\infty a Q\left(\!\sqrt{\frac{bM\Omegaul h}{\sigma_n^2}}\right)\e^{-h} \dd h=\frac{a}{2}\left[\!1\!-\!\sqrt{\frac{b\bar{\gamma}\delta}{2+b\bar{\gamma}\delta}}\right],\vspace{-0.1cm}
\label{eq:AER}
\end{equation}
where $\bar{\gamma}$ is defined as $\bar{\gamma}=\Omegaul \bar{X}/\sigma_n^2=\Omegaul/\left(\lambda \sigma_n^2\right)$.
\subsection{Outage Probability Analysis}
\label{ss:Outage_probability_analysis}
Since the CSI is unknown at the EH node, the node transmits data at a constant rate $R_0$ in bits/(channel use). Therefore, assuming a capacity-achieving code, a channel outage occurs whenever $R_0>\log_2(1+\gamma)\Rightarrow \gamma<\gamma_{\rm thr}$, where $\gamma$ is the UL instantaneous SNR and $\gamma_{\rm thr}=2^{R_0}-1$. Hence, the channel outage probability is
\begin{equation}
P_{\text{out}}\big|_{\rm channel}=\mathbb{P}\left(\gamma<\gamma_{\text{thr}}\right)\!=\!\mathbb{P}\left(\frac{M\Omegaul h}{\sigma_n^2}\!<\!\gamma_{\text{thr}}\right)\!=\!1-\e^{-\frac{\gamma_{\text{thr}}}{\delta\bar{\gamma}}}.
\label{eq:Pout_infinite}\vspace{-0.1cm}
\end{equation}
Define the total outage probability as the probability that either the EH node loses a transmission opportunity when not enough energy is available in its energy buffer, or the node transmits in the UL but a channel outage occurs. Hence, the total outage probability is given by
\begin{equation}
P_{\text{out}}\big|_{\rm total}=\left(1-\mathbb{P}_{\rm trans}\right)+\mathbb{P}_{\rm trans}P_{\text{out}}\big|_{\rm channel}.
\label{eq:Total_outage_Finite}
\end{equation}
Since $\mathbb{P}_{\rm trans}$ is independent of the SNR, the diversity order of the total outage probability is solely governed by that the channel outage probability. Using $\lim_{y\to\infty}\e^{-\frac{1}{y}}=1-\frac{1}{y}+o(y^{-2})$, the channel outage probability tends asymptotically to $P_{\text{out}}\big|_{\rm channel}\asymp\frac{\gamma_{\text{thr}}}{\delta\bar{\gamma}}$, i.e., with a diversity order of 1. Hence, the diversity order of the total outage probability assumes the maximum possible value of one independent of the capacity of the energy buffer.
\section{Simulation and Numerical Results}
\label{s:Simulations}
In this section, we evaluate the performance of the investigated energy management policy through simulations. The simulation parameters are listed in Table \ref{tab:simulation_parameters}. We sweep over $\tilde{\delta}=\tilde{M}/\widetilde{\bar{X}}=0.1,\ldots,1.5$, which corresponds to a desired UL transmit power of $M=0.6\,\mu{\rm W},\cdots,10.75\,\mu$W. 
\begin{table}[!tp]
\caption{Simulation Parameters}
\begin{tabular}{@{}ll@{}}  \toprule   
Parameter & Value \\ \midrule \addlinespace[0.5em]
AP to EH node distance & $5\,$m\\ \addlinespace[0.5em]
AP and EH node antenna gains & $12\,$dBi and $2\,$dBi\\ \addlinespace[0.5em]
DL transmit power & $\Pdl=1\,$W\\ \addlinespace[0.5em]
AP noise figure & $5\,$dB \\ \addlinespace[0.5em]
Bandwidth & $5\,$MHz \\ \addlinespace[0.5em]
Noise power & $\sigma_n^2=-103\,$dBm \\ \addlinespace[0.5em]
Path loss exponent of DL and UL channels & $2.7$ \\ \addlinespace[0.5em]
DL and UL channel models & Rayleigh block fading \\ \addlinespace[0.5em]
DL and UL center frequencies  & $915\,$MHz and $2.45\,$GHz\\ \addlinespace[0.5em]
RF-to-DC conversion efficiency & $\eta=0.7$\\ \addlinespace[0.5em]
Power amplifier inefficiency & $\alpha=1.5$ \\ \addlinespace[0.5em]
Storage efficiency & $\beta=0.9$ \\ \addlinespace[0.5em]
Average harvested energy & $\widetilde{\bar{X}}=\beta\bar{X}=10^{-5}\,$J\\ \addlinespace[0.5em]
SNR &  $\widetilde{\bar{\gamma}}=\Omegaul \widetilde{\bar{X}}/\sigma_n^2=24.6\,$dB \\ \addlinespace[0.5em]
Constant power consumption & $\Pc=0.2\,\mu$W\\ \addlinespace[0.5em]
Energy storage capacity & $K\!=\!4\tilde{M}$, $7\tilde{M}$, and $20\tilde{M}$\\ \addlinespace[0.5em]
\bottomrule
\end{tabular}
\label{tab:simulation_parameters}
\end{table}

\begin{figure}[!tp]
\centering
\includegraphics[width=0.47\textwidth, trim=0.2cm 0 0 0, clip]{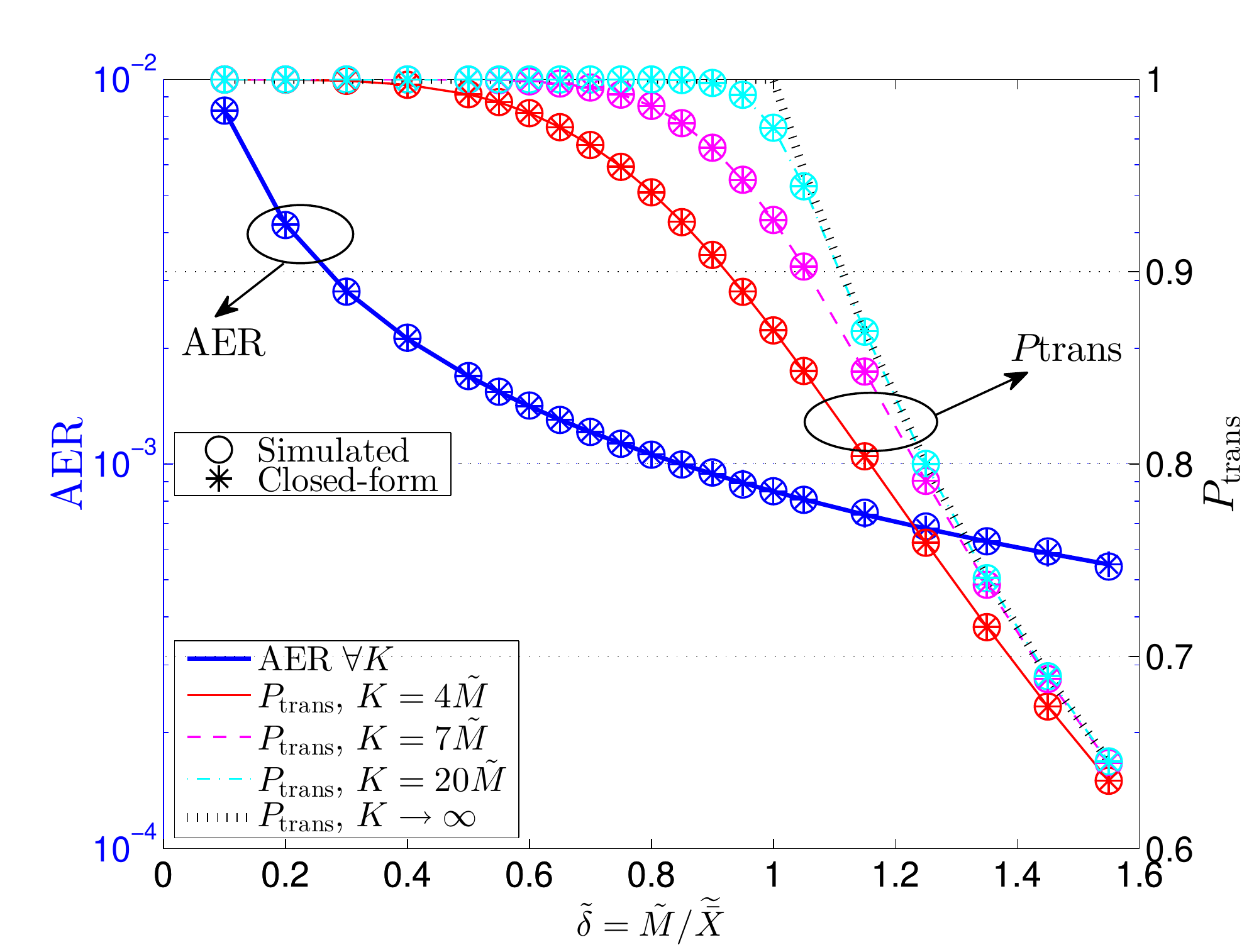}\vspace{-0.3cm}
\caption{AER and transmission probability for different buffer sizes and different desired UL transmit powers.}
\label{fig:BER_Rayleigh_finite_battery_sim_closedform_K_M}
\end{figure}
 
\begin{figure}[!tp]
\centering
\includegraphics[width=0.43\textwidth]{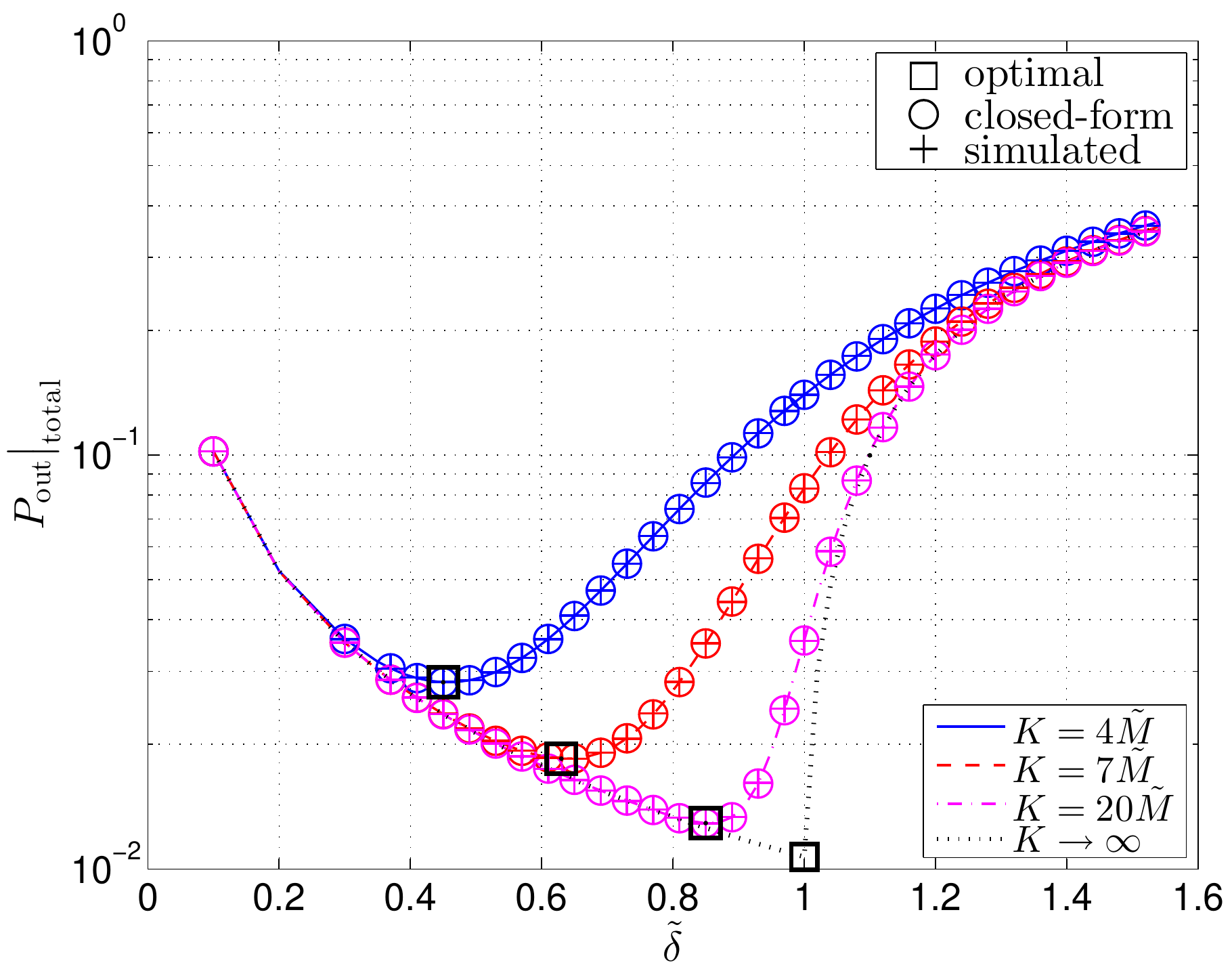}\vspace{-0.3cm}
\caption{Total outage probability for different buffer sizes and different SNRs.}
\label{fig:Outage_finite_sim_closed_form_K_M}
\end{figure}
Fig. \ref{fig:BER_Rayleigh_finite_battery_sim_closedform_K_M} shows the transmission probability together with the AER of the received signal at the AP when the EH node transmits a BPSK signal. The closed-form results shown in Fig. \ref{fig:BER_Rayleigh_finite_battery_sim_closedform_K_M} are obtained from the expressions in Section \ref{s:BER_outage_analysis}. We observe that the closed-form results agree perfectly with the simulated results. This emphasizes the tightness of the approximate pdf provided in Proposition \ref{prop:limiting_dist_Finite_exponential_approx}. Since the transmit power is constant whenever transmission is allowed, the AER is independent of the buffer size. However, a higher transmission probability is observed for larger energy buffer sizes. Furthermore, the larger the target transmit power $M$ (and therefore $\tilde{\delta}$), the lower the transmission probability. Thus, to achieve a certain desired maximum AER and a certain required transmission probability, the transmit power $M$ and the buffer size $K$ can be selected accordingly\footnote{We note that while instantaneous CSI knowledge is not required for the adopted transmission protocol, statistical CSI is needed if $M$ and $K$ are to be optimized.}.  

Fig. \ref{fig:Outage_finite_sim_closed_form_K_M} shows the total outage probability when the EH node transmits at a constant rate of $2.1\,$bits/(channel use), i.e., for $\gamma_{\rm thr}=5\,$dB. The total outage probability in (\ref{eq:Total_outage_Finite})  agrees perfectly with the simulated outage probability. For the considered SNR and $\gamma_{\rm thr}$, it is observed that the optimal $\tilde{\delta}$, for which the total outage probability  is minimized, is always $\leq1$ and  increases with the storage capacity (as $K\to\infty$, the optimal $\tilde{\delta}\to1$). In other words, the optimal desired UL transmit power is higher for larger energy buffers, but it is always less than the average harvested power. However, additional results --not shown here-- reveal that in the high outage regime ($P_{\text{out}}\big|_{\rm total}>0.5$), the opposite  behavior was observed. 
\section{Conclusions}
\label{s:conclusion}
We analyzed an on-off transmission policy for an EH node with finite/infinite energy storage.  Using the theory of discrete-time Markov chains on a general state space, the limiting distribution of the stored energy in the buffer was analyzed for a general i.i.d. EH process and obtained in closed form for an exponential EH process. An approximate limiting distribution was proposed for finite-size buffers and shown to be tight. Based on the performance analysis, the UL transmit power and the energy buffer size can be designed to achieve a desired AER/ channel outage probability for a given required transmission probability. Our results revealed that, except for high outage probability $(>0.5)$, the optimal desired transmit power of the EH node was always less than the average harvested power and increased with the storage capacity.

\appendices  
\section{Proof of Corollary \ref{theo:stationary_dist_infinite_exp}}
\label{app:stationary_dist_infinite_exp}
Define $g_1(x)=g(x)$, $0<x\leq M$ and $g_2(x)=g(x)$, $x>M$. Then, substituting $f(x)\!=\!\lambda\e^{-\lambda x}$ in (\ref{eq:g_integral_eqn_infinite}), we get\vspace{-0.1cm} 
\begin{align}
g_1(x)\!&=\!\!\!\int\limits_0^x\!\!\lambda\e^{-\lambda (x-u)} g_1(u) \dd u \!+\!\!\!\!\!\!\int\limits_{M}^{M+x}\!\!\!\!\lambda\e^{-\lambda (x-u+M)} g_2(u) \dd u\label{eq:integral_eqn_INF_battery_exp_input_parta}\\
g_2(x)\!&=\!\!\!\int\limits_0^M\!\!\lambda\e^{-\lambda (x-u)} g_1(u) \dd u \!+\!\!\!\!\!\!\int\limits_{M}^{M+x}\!\!\!\!\lambda\e^{-\lambda (x-u+M)} g_2(u) \dd u.\label{eq:integral_eqn_INF_battery_exp_input_partb}
\end{align}
When $M>\bar{X}$, i.e., $\delta>1$, we know from Theorem \ref{theo:stationary_dist_infinite} that $g_1(x)$ and $g_2(x)$ have unique solutions. We postulate an exponential-type solution for $g_2(x)$ given by $g_2(x)=k\e^{px}$, which is required (denoted by $\req$) to satisfy (\ref{eq:integral_eqn_INF_battery_exp_input_partb}), i.e.,
\begin{equation}
k\e^{px}\req \e^{-\lambda x}\left[\lambda\int\limits_0^M\e^{\lambda u}g_1(u)\dd u-k\frac{\lambda \e^{pM}}{\lambda+p}\right]+\frac{\lambda \e^{pM}}{\lambda+p} k\e^{px}.
\label{eq:infinite_exp_g2}
\end{equation} 
For the postulated $g_2(x)$ to be correct, (\ref{eq:infinite_exp_g2}) must hold, i.e., the coefficient of $\e^{px}$ in the last term of (\ref{eq:infinite_exp_g2}) must be $k$, which implies $\lambda\e^{pM}\!=\!\lambda\!+\!p$. This condition should also reduce the coefficient of $\e^{-\lambda x}$ to zero, i.e., $\int_0^M\e^{\lambda u}g_1(u)\dd u=\frac{k}{\lambda}$ must also hold. From $\lambda\e^{pM}\!=\!\lambda+p$, $p$ can be obtained using the Lambert W function, i.e., $p\!=\!\left(-\delta-W_{0}(-\delta\e^{-\delta})\right)/M$, which is $<0$ since $\delta>1$. Now, to get $g_1(x)$, we substitute $g_2(x)=k\e^{px}$ in (\ref{eq:integral_eqn_INF_battery_exp_input_parta}) and use $\lambda\e^{pM}\!=\!\lambda\!+\!p$ to get
\begin{equation}
g_1(x)=\int\limits_0^x\lambda\e^{-\lambda (x-u)} g_1(u) \dd u +\underbrace{k\left(e^{px}-\e^{-\lambda x}\right)}_{r(x)}.
\label{eq:g1_step}
\end{equation}
Eq. (\ref{eq:g1_step}) is a Volterra integral equation of the second kind, whose solution is given by \cite[eq. 2.2.1]{polyanin2008handbook}
\begin{equation}
g_1(x)=r(x)+\lambda\int\limits_0^x r(t)\dd t=k\frac{\lambda}{p}\e^{pM}\left[\e^{px}-1\right].
\label{eq:g1_step2}
\end{equation}
Now, $k$ can be obtained from the unit area condition on $g(x)$, namely, $\int_0^M g_1(x)+\int_M^\infty g_2(x)\!\!=\!\!1 \Rightarrow k=\frac{-p}{M(\lambda+p)}=\frac{-p}{\delta\e^{pM}}$. Substituting $k$ back in (\ref{eq:g1_step2}), we get $g_1(x)=\frac{1}{M}\left(1-\e^{px}\right)$. Now, it can be shown that $g_1(x)$ satisfies $\int_0^M\e^{\lambda u}g_1(u)\dd u=\frac{k}{\lambda}$. Hence, the postulated $g_2(x)=k\e^{px}$ and the resulting $g_1(x)=\frac{1}{M}\left(1-\e^{px}\right)$ are the unique solutions for the coupled equations in (\ref{eq:integral_eqn_INF_battery_exp_input_parta}) and (\ref{eq:integral_eqn_INF_battery_exp_input_partb}). Finally, the transmission probability can be obtained as $\mathbb{P}(\Pul(i)=M)=\mathbb{P}(B(i)>M)=1-\int_{0}^{M} g_1(x)\dd x=\frac{1}{\delta}$. This completes the proof.
\section{Proof of Proposition \ref{prop:limiting_dist_Finite_exponential_approx}}
\label{app:limiting_dist_Finite_exponential_approx}
First, the approximate pdf $\tilde{g}(x)$ in the range $M\leq x<K-n_cM$ is identical to that in \cite[Proposition 1]{Morsi_ICC2014}. In particular, the proposed exponential approximation $\tilde{g}(x)=c\e^{dx}$ is motivated by the exponential distribution of the buffer content for an infinite buffer size in the range $x>M$ given in (\ref{eq:g_x_infinite_buffer}), c.f. Corollary \ref{theo:stationary_dist_infinite_exp}. With $\tilde{g}(x)=c\e^{d x}$, $d$ is obtained in exactly the same manner as $p$ for an infinite-size buffer, c.f. Appendix \ref{app:stationary_dist_infinite_exp}. However, unlike in the infinite-size buffer case,  the amount of energy in a finite-size buffer with $\delta \leq 1$ still convergences to a limiting distribution, c.f. Theorem \ref{theo:limiting_dist_Finite}. This explains the use of the Lambert W function with two different orders in (\ref{eq:d_approx}) to consider the two cases of $\delta \leq 1$ and $\delta>1$. Note that $d$ in (\ref{eq:d_approx}) satisfies $d>0$ for $\delta<1$ (an exponentially increasing $\tilde{g}(x)$), $d<0$ for $\delta>1$ (an exponentially decaying $\tilde{g}(x)$), and $d=0$ for $\delta=1$ (a uniformly distributed $\tilde{g}(x)$). Similar to \cite[Proposition 1]{Morsi_ICC2014}, we obtain $c$ from the exact $g(x)$ in (\ref{eq:g_x_finite_exact_partb}) at $x=0$ and $n=l-1$ after replacing $\pi(K)$ by $\tilde{\pi}(K)$. Consider next the approximate pdf $\tilde{g}_{l-1}(x)$ in the range $0\leq x<M$. $\tilde{g}_{l-1}(x)$ is obtained in exactly the same manner as $g_1(x)$ in the infinite-size buffer case, c.f. Appendix \ref{app:stationary_dist_infinite_exp}. In particular, in (\ref{eq:integral_eqn_INF_battery_exp_input_parta}), we use $\tilde{g}(x)$ in place of $g_2(x)$ and $\tilde{g}_{l-1}(x)$ in place of $g_1(x)$ and obtain $\tilde{g}_{l-1}(x)$ by solving a Volterra integral equation of the second kind. Hence, $\tilde{g}_{l-1}(x)$ reduces to (\ref{eq:g1_step2}) after replacing $k$ with $c$ and $p$ with $d$.
Next, in the range $k-n_cM \leq x < K$, we use the exact $g(x)$ in (\ref{eq:g_x_finite_exact_partb}) after replacing $\pi(K)$ by $\tilde{\pi}(K)$. The reason why we do not use $\tilde{g}(x)=c\e^{dx}$ to approximate  the pdf in the whole range of $M\leq x<K$ is that although the approximate pdf $\tilde{g}(x)$ is tight for most of this range (even with $n_c=2$), it is loose at the tail of the distribution (namely for the last two sections of the pdf, i.e., $n=0,1$). As far as the performance analysis is concerned, only the pdf in the range $[0,M]$ is needed, c.f. Section \ref{s:BER_outage_analysis}.
Finally, the approximate full buffer probability $\tilde{\pi}(K)$ in (\ref{eq:Pi_K_approx}) guarantees a unit area distribution, i.e.,
\begin{equation}
\int\limits_{0}^{M} \tilde{g}_{l-1}(x) \dd x +\!\!\!\int\limits_{M}^{K-n_c M}\!\!\!\!\!\!\tilde{g}(x) \dd x+ \sum\limits_{n=0}^{n_c-1}\int\limits_{K-(n+1)M}^{K-nM}\!\!\!\!\!\!\!\tilde{g}_n(x) \dd x+\tilde{\pi}(K)=1.
\label{eq:unit_area_pdf_approx_only_output_M_step1}
\end{equation}
Next, we study the error associated with the proposed approximation. Using (\ref{eq:g_x_finite_exact_parta}) and (\ref{eq:gtilde_l_1}), the approximation error in the range $[0,M]$, $e(x)=g_{l-1}(x)-\tilde{g}_{l-1}(x)$, is given by
\small
\begin{equation}
\begin{aligned}
&e(x)\!=\!\pi(K)\lambda\e^{\delta(l-1)}\!\left[R(1-l,l-2)-\e^{-\lambda x}R(1-l+\frac{x}{M},l-2)\right]\\
&-\tilde{\pi}(K)\lambda\e^{\delta l}\left[R(-l,l-1)-\e^{-\delta}R(1-l,l-2)\right]\frac{\lambda\e^{dM}}{d}(\e^{dx}-1),
\end{aligned}
\label{eq:approx_error}
\end{equation}
\normalsize
where we define $R(y,l)=\sum_{q=0}^{l}(y+q)^q\frac{(\delta\e^{-\delta})^q}{q!}$. Using (\ref{eq:approx_error}) at $n_c=2$, the maximum error percentage $e(x)/g(x)$ in the range of $x=[0,M]$ with $\delta\geq 0.5$ is less than 8.3\% for $K=3M$ and 1.64\% for $K=4M$. Next, we show that the  approximation error tends to zero as the buffer size tends to infinity (for $\delta>1$, otherwise, a stationary distribution does not exist, c.f. Theorem  \ref{theo:no_stationary_dist}). Using the asymptotic expansion of the exponential of the Lambert W function given by $\frac{\e^{-aW_j(-z)}}{1+W_j(-z)}=\sum_{q=0}^{\infty}(a+q)^q\frac{z^q}{q!}$, it can be shown that $g_{l-1}(x)$ in (\ref{eq:g_x_finite_exact_parta}) and $\tilde{g}_{l-1}(x)$ in (\ref{eq:gtilde_l_1}) tend asymptotically (as $l\to\infty$) to
\small
\begin{equation}
\lim\limits_{l\to\infty}g_{l-1}(x)=\lim\limits_{l\to\infty} \pi(K)A(l,\delta,\lambda)\left(1-\e^{-\left(\lambda+\frac{W_j(-\delta\e^{-\delta})}{M}\right)x }\right)
\label{eq:g_l_1_asymp}
\end{equation}
\begin{equation}
\lim\limits_{l\to\infty}\tilde{g}_{l-1}(x)=\lim\limits_{l\to\infty} \tilde{\pi}(K)A(l,\delta,\lambda)\left(1-\e^{-\left(\lambda+\frac{W_j(-\delta\e^{-\delta})}{M}\right)x }\right),
\label{eq:gtilde_l_1_asymp}
\end{equation}
\normalsize
where $A(l,\delta,\lambda)=\left(\lambda\e^{(l-1)(\delta+W_j(-\delta\e^{-\delta}))}\right)/\left(1+W_j(-\delta\e^{-\delta})\right)$. Hence, (\ref{eq:g_l_1_asymp}) and (\ref{eq:gtilde_l_1_asymp}) differ only in the atom value. Now, using $\pi(K)$ in (\ref{eq:PiK_exact_Transmit_M_only_step}) and $\tilde{\pi}(K)$ in (\ref{eq:unit_area_pdf_approx_only_output_M_step1}) together with $\lim_{l\to\infty}\pi(K)=\lim_{l\to\infty}\tilde{\pi}(K)=0$ (i.e., the atom vanishes as $K\to\infty$), it can be shown that
\begin{equation}
\lim\limits_{l\to\infty} \pi(K)A(l,\delta,\lambda)=\lim\limits_{l\to\infty}\tilde{\pi}(K)A(l,\delta,\lambda)=\frac{1}{M}.
\label{eq:common_term_with_pi_and_e_power_l}
\end{equation}
Consequently,
\begin{equation}
\begin{aligned}
\lim\limits_{l\to\infty}g_{l-1}(x)&=\lim\limits_{l\to\infty}\tilde{g}_{l-1}(x)=\frac{1}{M}\left(1-\e^{-\left(\lambda+\frac{W_j(-\delta\e^{-\delta})}{M}\right)x }\right)\\
&=\frac{1}{M}\left(1-\e^{px}\right),
\end{aligned}
\label{eq:g_l_1_equivalence_to_inf_buffer}
\end{equation}
as obtained in Corollary \ref{theo:stationary_dist_infinite_exp} and the approximation error tends to zero. This completes the proof.
\bibliographystyle{IEEEtran}
\bibliography{references}
\end{document}